\documentclass[12pt]{article}
\usepackage{amsmath}
\usepackage{amsfonts}
\usepackage{amssymb}
\usepackage{amsthm}
\usepackage{stackrel}
\usepackage{color}
\usepackage[T2A]{fontenc}
\usepackage[T2A]{fontenc}
\usepackage[utf8]{inputenc}
\usepackage{epsfig}
\usepackage{graphics}
\usepackage{epsf}
\usepackage{figsize}

\newtheorem{thm}{Theorem}[section]

\newcommand{\R}{{\rm I}\kern-0.18em{\rm R}}
\newcommand{\1}{{\rm 1}\kern-0.25em{\rm I}}
\newcommand{\E}{{\rm I}\kern-0.18em{\rm E}}
\newcommand{\p}{{\rm I}\kern-0.18em{\rm P}}

\title{Heavy-tailed probability distributions in social sciences}
\author{Lev B. Klebanov\thanks{Department of Probability and Mathematical Statistics, Charles University, 
Prague, Czech Republic. e-mail: levbkl@gmail.com}, Yulia V. Kuvaeva\thanks{ Department of Finance, Money Circulation and Credit, Ural State University of Economics,
		620144 Yekaterinburg, Russia; ykuvaeva1974@mail.ru}}
\date{}

\begin{document}
\maketitle

\begin{abstract}
We present an overview of possible reasons for the appearance of heavy-tailed distributions in applications 
to the natural sciences. These distributions include the laws of Pareto, Lotka, and some new ones. 
The reasons are illustrated using suitable toy models.

\vspace{0.2cm}
\end{abstract}

{\bf Epigraph}

\vspace{0.25cm}
King James Bible:

\vspace{0.2cm}
{\it I returned, and saw under the sun, that the race is not to the swift, nor the battle to the strong, neither yet bread to the wise, nor yet riches to men of understanding, nor yet favour to men of skill; but time and chance happeneth to them all}.

\vspace{0.25cm}
Lion Feuchtwanger "Der jüdische Krieg":

\vspace{0.2cm}
{\it Gods give talent to one and luck to another}.

\section{History of the problems}\label{sec1}
\setcounter{equation}{0}

The distributions with heavy (power-like) tails in social sciences were used more than hundred years.

Relevant studies include the following:
\begin{itemize}
\item[1.] Distribution of big capital. Pareto, 1896 (see \cite{Pa})). The density is $p(x)=\frac{\alpha}{x_o}\Bigl(\frac{x_o}{x}\Bigr)^{\alpha+1}$, for $x\geq x_o$, $\alpha>0$.
\item[2.] Scientific production. The number of scientists who published one, two and so on $x$ papers ( the number $x$ published by scientist papers). Lotka (1926) (see \cite{Lo}) showed that $n(x) = n_1/x^a$, where $n_1>0$, $a \leq 2$  (in many cases  $a \approx 2$).  
\item[3.] Lotka's law approximately holds for the number of citations of a paper by a scientist. 
\item[4.] For a specific artistic text, the sequence of all words is written in descending order according to the frequency of their occurrence. Comparing the frequency of the word and the place in this sequence (rank) leads to  $x=B/r$, $B = const$ (see \cite{Zi}).
\end{itemize}

Why do these patterns emerge? 
Probably, Laws 1 - 3 refer to some individual human abilities, while Law 4 refers to memory or other functions of the human brain.

\vspace{0.2cm}
We will not consider the 4th law in the paper and will focus on laws 1 and 3, more precisely on their qualitative explanation. It is so because Zipf explained his Law based on the least effort principle. Although there are no rigorous results on the existence of a mechanism related to this principle in the human brain, not wasting memory seems natural.
However, the application of the the least effort principle in cases 1-3 does not seem to be related to the essence of the issues under consideration.

\vspace{0.2cm}
At first glance, everything looks quite simple. The population of a country is heterogeneous. There are people more capable of business (in case 1) or scientific work (in cases 2, 3) and people who are not (or less capable) of such activities.

\vspace{0.2cm}
But how big is the difference in ability, and are all differences in "success" determined by ability?

\vspace{0.2cm}
Let's remember the epigraph!

\vspace{0.2cm}
Is there an effect of chance?
First, let's focus on the 1st law. Let's try to build a model that explains the reason for its occurrence.

However, the income and the capital distribution are subject to many factors not fully accounted for. Our interest is not in the whole mechanism of accumulation and distribution of capital but in the roles of human talent and chance in this process only. How essential are these roles? Therefore, we have to use a toy model which assumes all people have identical abilities. If the role of chance is small, then there will not be many variations in the model between different investors. In contradiction, if we see a large difference between investors, this will indicate a significant role of chance.

\section{Toy model for capital distribution}\label{sec2}
\setcounter{equation}{0} 

Let us consider the first toy model of capital distribution leading to the Pareto law. 

Suppose for simplicity that there exists only one business. All possible investors are equal in their talents and initial capital. Consider the case when each investor invests one unit of capital in the business. After one time-unit, the business outcome is $X_1$, where $X_1$ is a random variable. Suppose the investor left all this sum into business, and the conditions on the market remain the same during the following time interval. Then the outcome after the second-time interval is $X_1\cdot X_2$, where $X_1$ and $X_2$ are independent identically distributed (i.i.d.) random variables. In the same way, the outcome after the n-th time interval is $\prod_{i=1}^{n}X_ j$, and $X_1, X_2, \ldots, X_n$ are i.i.d. random variables. Let us suppose that the conditions on the market will change radically at a random moment $\nu_p$ so that investing in that business becomes not profitable.
Therefore, the final outcome is $\prod_{j=1}^{\nu_p}X_j$. We are interested in the outcome behavior for large $\nu_p$ values.
More precisely, we suppose that 
\begin{enumerate}
\item $\mathbf{X}=\{X_1, X_2, \ldots , X_n, \ldots \}$ is a sequence of i.i.d. positive random variables, $a=\E \log X_1$;
\item $\nu=\{ \nu_p, p\in \Delta \subset (0,1)\}$ is a family of positive integer-valued random variables independent with the sequence $\mathbf{X}$, $\E \nu_p =1/p$.

\vspace{0.2 cm}
Generally, the information on the $\nu$--family is unavailable.  We shall consider a few cases starting with a simple one.
\item $\p\{\nu_p=k\}=p\cdot (1-p)^{k-1}$, $k=1,2,\ldots$ i.e. $\nu_p$ has geometric distribution.
\end{enumerate}

Define $Z_p=\prod_{j=1}^{\nu_p}X_j^p$.
\vspace{-0.2cm}
\begin{thm}\label{th1}
\vspace{-0.1cm}
Suppose that the items 1. -- 3. hold. Let $a\neq 0$. Then
\vspace{-0.2cm}
\[ \lim_{p \to 0}\p\{Z_p<x\} =1-x^{-1/a}, \; \text{for}\;x \geq 1,\; a>0 \]
\vspace{-0.2cm}
and
\[ \lim_{p \to 0}\p\{Z_p<x\} =x^{1/a}, \; \text{for}\;x \leq 1,\; a<0 .\]
\end{thm}
\vspace{-0.25cm}
In the case of $a>0$ (profitable business), we have a Pareto distribution which Wilfredo Pareto had proposed on the basis of empirical study (see \cite{Pa}).  For the proof of 
Theorem \ref{th1} see \cite{KMR}. In the paper, \cite{KMR} is obtained the result for $a=0$. For this case $Z_p$ must be changed by $Z_p^{\prime}=\prod_{j=1}^{\nu_p}X_j^{\sqrt{p}}$. Under the condition of the existence of logarithmic second moment of $X_1$ the product $Z_p^{\prime}$ converges in distribution to the mixture of the given in Theorem \ref{th1} distributions. It is well-known that Pareto distribution has heavy tails. This fact implies that capital belongs to a relatively small number of people. Now we see Pareto distribution appears in a very natural way described as a limit distribution for a product of a random number $\nu_p$ of random variables $X_j$. The value of $\nu_p$, $p \in (0, 1)$ in 3. had a geometric distribution. What will happen with others ("natural") distributions? Below we consider two additional cases:
\begin{enumerate}
\item[4.] $\nu_p$ has a probability generating function 
\[ \mathcal{P}(z,p,m) =\frac{p^{1/m}z }{(1-(1-p)z^m))^{1/m}}, \quad p\in (0,1), \quad m \in \mathbb{N}.\]
\item[5.] $\nu_p$ has a probability generating function 
\[ \mathcal{P}(z,n) =\frac{1}{T_n(1/z)}, \]
where $T_n(u)$ is Chebyshev polynomial of the first kind and $n=1/\sqrt{p}$ is its degree. $\E \nu_p = 1/p$.
\end{enumerate}

Let us consider case 4. The following result holds.
\begin{thm}\label{th2} Suppose that the items 1., 2. and 4. hold. Let $a\neq 0$. Then
\[  \lim_{p \to 0}\p\{Z_p<x\} =\int_{1}^{x}\frac{1}{b^{1/m}\Gamma(1/m)u^{1+1/m}\log^{1-1/m}(u)}du, \; \text{for}\;x \geq 1, \]
where $b>0$ is a parameter.
\end{thm}
\begin{proof} 
Consider $\log Z_p = p \sum_{j=1}^{\nu_p} Y_j$, where $Y_j = \log X_j$. From the result of the paper \cite{Mel} it follows that the limit distribution of $\log Z_p$ as $p \to 0$ has the density 
$ \exp\{-u/b\}/\bigl(u^{1-1/m}b^{1/m}\Gamma(1/m)\bigr)$, $u>0$. Now it is sufficient to pass the limit distribution of $Z_p$ from its logarithm density.
\end{proof}

\begin{thm}\label{th3}
Suppose that the items 1., 2. and 5. hold. Let $a= 0$ and the second logarithmic moment of $X_1$ exists. Then
\[  \lim_{p \to 0}\p\{Z_p^{\prime}<x\} = \frac{2}{\pi}\arctan(x^b), \quad \text{for}\quad x>0,\]
where $b>0$ is a parameter.
\end{thm}
\begin{proof}
Similarly to the proof of the previous Theorem, we have to pass from $Z_p^{\prime}$ to its logarithm, apply the corresponding result from \cite{KKRT}, and went back to the limit distribution for initial random variables.  
\end{proof}

All three models constructed above do not take into account any abilities of the investing in the given enterprise people but lead to heavy-tailed distributions. The difference between investors is only in the occurrence of some unfavorable event for them (moment $\nu_p$). Objection that this moment is the same for the whole store, i.e. it is insolvent for all investors at once because the investors invested in the business at different times. Therefore, the period for which the investment was made is different for each investor. So we see that the dependence on the moment and the case are really very high. I do not deny that the dependence on the talent of the investor is indeed significant, but it would be very difficult to separate this component from random.

\section{Citations distribution}\label{sec3}
\setcounter{equation}{0} 

A similar situation occurs when studying the distribution of citations of scientific publications. Let us make some Assumptions.

\vspace{0.2cm}
\textbf{Assumption 1.}

{\it All scientists under consideration are equal in their scientific and literary abilities.}

\vspace{0.2cm}
\textbf{Assumption 2.}

{ \it Paper citations occur independently.}

\vspace{0.2cm}
\textbf{Assumption 3.}
The probability that an article will be repeatedly cited depends on the number of previous citations. It is growing with the growing citations number. More precisely, 

{\it Assuming the probability that an article having $\mathit{k-1(k\geq 1)}$ citations will have no further citations is }
\begin{equation*}
	p_{k}=\frac{1}{(a\,k+b)},
\end{equation*}
{\it where$a>0$ a $b\geq 0$ are real numbers satisfying  $a+b>1$}.

Let $Y$ be a random variable describing the number of citations during the considered period. Assumption 1 implicitly de-facto implies that $Y$ has the same distribution for different papers because the scientific abilities of the authors are supposed to be the same.

In view of citation independence, the probability that a paper is cited exactly $n$ times is
\begin{equation*}
	\mathrm{I}\kern-0.18em\mathrm{P}\{Y=n\}=p_{n}\prod_{k=1}^{n-1}(1-p_{k})=
	\frac{\Bigl(\frac{a+b-1}{a}\Bigr)_{n-1}}{(a\,n+b)\Bigl(\frac{a+b}{a} \Bigr)_{n-1}},
\end{equation*}
where $(a)_{n}=a(a+1)\ldots (a+n-1)$ is  Pochhammer symbol.

It is not difficult to calculate the probability
\begin{equation} \label{eq1}
	\mathrm{I}\kern-0.18em\mathrm{P}\{Y\geq m\}=\frac{\Bigl(\frac{a+b-1}{a}\Bigr)_{m-1}}{\Bigl(
		\frac{a+b}{a}\Bigr)_{m-1}} \stackrel[m \to \infty]{}{\thicksim} \frac{\Gamma((a+b)/a)}{\Gamma((a+b-1)/a)}\frac{1}{m^{1/a}} \; .
\end{equation}

Distribution of citations
The relation (\ref{eq1}) shows that the distribution of citations has a heavy tail, the severity of which depends on the value of the parameter $a$ responsible for the degree of influence of previous citations. Therefore, a larger value of $a$ corresponds to a heavier tail. In any case, the presence of such a tail makes it possible to conclude that the citation intensity of almost identical scientists can differ significantly, which leads to a significant stratification of the scientific community through various random circumstances that have nothing to do with research abilities. Thus, the citation number seems meaningless as an indicator of scientific value.

{\bf Make now some remarks on the Impact Factor distribution.}

Let us now consider the possibility of using the impact factor of a journal as an indicator of the scientific significance of a paper published in it. The impact factor of a journal is calculated as the ratio of the number of citations of papers published over a certain period to the number of these papers themselves. The idea of considering such an average value is connected with the idea that, according to the law of large numbers, the influence of chance will be leveled. However, we shall show, this is not true.

Mention that there exists a rather large literature stating the scientific journals' impact factor has essential value. Based on the observed data, the presence of asymmetry in the distribution of the impact factor and the presence of a heavy tail is noted. However, these circumstances are not analyzed from a theoretical point of view, and only comments are made on the advisability of replacing the arithmetic mean with some other statistics for the purpose of statistical data analysis. We note one of the typical works of this kind \cite{Bl}. True, the author notes the similarity of the distribution of some data with the Pareto distribution, but the mathematical analysis of the reasons for its occurrence is not carried out. In addition, not a mathematically strictly defined distribution is considered, but only its "naive" form. Below we will try to clarify the appearance of heavy tails of the impact factor distribution.

We assume that the number of papers submitted to the journal has a Poisson distribution. For simplicity, let us assume that the number of citations for each of the submitted papers has a Sibuya distribution. Then the citation distribution of all papers has a probability generating function that is a superposition of the generating functions of the Sibuya and Poisson laws.
Probability generating function of this superposition is $\mathcal P (z)= e^{ -\lambda (1-z )^{p}}$ for fixed $\lambda>0$ and $p\in (0,1)$. Clearly, this distribution has a heavy tail with index $p$. In view of the fact that $p <1$ the law of large numbers is inapplicable in this situation. Moreover, in this case, the impact factor increases with the number of publications without increasing their scientific significance. The observed increase (over time) in the impact factors of leading journals confirms this circumstance.

Now we can conclude the impact factor distribution has a heavy tail again and cannot be used as an indicator of scientific significance.

\section{Conclusions}\label{sec4} 
\setcounter{equation}{0}

\begin{itemize}
	\item[I.] It is shown that distributions with heavy tails can arise in some manifestations of social inequality (the distribution of capital, the number of citations, the impact factor) due to purely random reasons. In this case, the spread in the magnitude of inequality is significant.
	\item[II.]  The circumstance specified in 1. makes it impossible to use such indices as the number of citations and/or the impact factor of a journal as an indicator of the scientific significance (scientific quality) of a published work.
\end{itemize}

\section*{Acknowledgment}\label{secA} 
\setcounter{equation}{0} 
The work by Lev B. Klebanov was partially supported by GA \v{C}R Grant
19-28231X EXPRO.

\end{document}